\documentclass[aps,prl,twocolumn,showpacs,showkeys,a4]{revtex4}
\usepackage{graphicx}
\usepackage{amsmath}
\usepackage{amssymb}

\usepackage[english]{babel}
\usepackage{dsfont}
\usepackage{bbm}
\usepackage{hyperref} 
\usepackage{amsmath, amsthm, amssymb, mathrsfs}

\usepackage[pdftex,usenames]{color}

\newtheorem{lemma}{Lemma}
\newtheorem{theorem}{Theorem}

\DeclareMathOperator{\Tr}{\mathrm{Tr}}
\DeclareMathOperator{\probability}{\mathrm{Pr}}
\DeclareMathOperator{\tracedistance}{\mathcal{D}}
\DeclareMathOperator{\ee}{\mathrm{e}}
\DeclareMathOperator{\iu}{\mathrm{i}}
\DeclareMathOperator{\hiH}{\mathcal{H}}
\DeclareMathOperator{\haH}{\mathscr{H}}

\newcommand{\bra}[1]{\langle #1|}
\newcommand{\ket}[1]{|#1\rangle}

\newcommand{\ketbra}[2]{| #1 \rangle \langle #2 |}
\newcommand{\expect}[1]{\langle #1\rangle}

\begin{document}
\title{Einselection without pointer states}
\author{Christian Gogolin\footnote{publications@cgogolin.de}}
\affiliation{Fakult{\"{a}}t f\"{u}r Physik und Astronomie, Universit\"{a}t W\"{u}rzburg, Am Hubland, 97074 W\"{u}rzburg, Germany}

\begin{abstract}
  We investigate decoherence and equilibration in the experimentally relevant situation of weak coupling to an environment.
  We consider small subsystems of large, closed quantum systems that evolve according to the von~Neumann equation.
  Without approximations and without making any special assumptions on the form of the interaction we prove that, for almost all initial states and almost all times, the off-diagonal elements of the density matrix of the subsystem in the eigenbasis of its local Hamiltonian must be small, whenever the energy difference of the corresponding eigenstates is larger than the interaction energy.
  This proves that decoherence with respect to the local energy eigenbasis is a natural property of weakly interacting quantum systems.
\end{abstract}


\maketitle

\section{Introduction}
\label{sec:introduction}
Quantum Mechanics claims to be a fundamental theory.
As such, it should be able to provide us with a microscopic explanation for all phenomena we observe in macroscopic systems, including irreversible processes such as thermalization.
But its unitary time development seems to be incompatible with irreversibility, leading to an apparent contradiction between Quantum Mechanics and Thermodynamics \footnote{Note that this is not a genuine problem of Quantum Mechanics. Classical Mechanics is time reversal invariant as well.}.

To explain irreversible processes and to overcome the problem of the emergence of classicality many authors have suggested to modify Quantum Theory.
For example, by adding non-linear terms to the von~Neumann equation, or by postulating a periodic spontaneous collapse of the wave function \cite{Zeh96,Bassi03}.
Others have considered Markovian, nonunitary, time evolution \cite{Zeh96,Breuer02} which can be thought of as resulting from an interaction with a memoryless bath and it has been shown that system bath models that evolve under a special type of Hamiltonian tend to evolve into states that are classical superpositions of so called \emph{pointer states} \cite{RevModPhys.75.715} --- a phenomenon called \emph{einselection}.
These approaches, which are subsumed under the term \emph{decoherence theory}, are able to reproduce many of the features of dissipative systems and are undoubtedly very valuable for applications.

But, in face of the enormous success of standard Quantum Mechanics in explaining microscopic phenomena and the existence of macroscopic quantum systems on the one hand and the broad applicability of Statistical Mechanics and Thermodynamics on the other, we feel that neither a modification of Quantum Theory, nor considerations restricted to special Hamiltonians can provide a satisfactory explanation for the classical, statistical, and thermodynamic behavior of the macroscopic world.
Recently there has been remarkable progress in explaining macroscopic, seemingly irreversible behavior from standard Quantum Mechanics.
It has been shown that it is possible to explain the phenomenon of equilibration and irreversibility \cite{Reimann08,Linden09,0907.1267v1} and to justify the applicability of the canonical and microcanonical ensemble \cite{Popescu06} without added randomness (i.e., without assuming the existence of already equilibrated and thermalized baths) and ensemble averages, from nothing but pure Quantum Mechanics and the randomness due to entanglement with the environment (see also \cite{Gogolin10} and the references therein). 

We make use of the results obtained in these papers and connect this approach with the research on decoherence.
We consider the case of decoherence due to weak interaction with an environment.
A weak coupling to an environment exists in practically all situations.
This case is thus of  particular interest for developing a better understating of the foundations of Statistical Mechanics and Thermodynamics, for applications in quantum information processing and quantum computing and for experiments on environment-assisted entanglement creation (see \cite{Benatti06} and the references therein).
Our main result is that decoherence with respect to the local energy eigenbasis is a natural property of weakly coupled systems.

\section{Setup and notation}
\label{sec:setupanddefinitions}
We consider arbitrary quantum systems that can be described using a Hilbert space $\hiH$ of finite dimension $d$ and that can be divided into two parts, which we will call the bath $B$ and the subsystem $S$.
For infinite dimensional systems it is often possible to find an effective description in a finite dimensional Hilbert space by introducing a high energy cut-off.
Moreover, it was demonstrated in \cite{Devi09} that many of the phenomena that can be rigorously proven in the finite dimensional case also occur in infinite dimensional systems.
We thus believe that the restriction to finite dimensions as mainly a technicality.

We use the terms bath and subsystem because in the end we will be interested in situations where the dimension $d_B$ of the Hilbert space of the bath $\hiH_B$ is much larger than the dimension $d_S$ of the Hilbert space $\hiH_S$ of the subsystem, besides that $S$ and $B$ are two completely arbitrary quantum systems.

We assume that the Hamiltonian $\haH$ of the joint system has \emph{non-degenerate energy gaps}.
This assumption already appears in the work of von~Neumann~\cite{vonneumann1929} and later in \cite{Linden09,Reimann08} and means that for any four energy eigenvalues $E_k,E_l,E_m,E_n$ equality of the gaps $E_k - E_l = E_m - E_n$ implies that either $k=l$ and $m=n$ or $k=m$ and $l=n$.
It shall be emphasized that this is an extremely weak assumption and with some additional effort, it can be replaced by an even weaker one that allows degeneracies of the energy levels \cite{0907.1267v1}.
Every Hamiltonian becomes non-degenerate by adding an arbitrary small random perturbation; therefore the Hamiltonians of macroscopic systems can be expected to satisfy this constraint.
The physical implication of the above assumption is that the Hamiltonian is \emph{fully interactive} in the sense that there exists no partition of the system into two non interacting subsystems.
For fully interactive systems our results are robust against the existence of some degeneracies in the energy gaps.
How non-degenerate the energy spectrum is influences the equilibration and decoherence times.

We use $\rho$ for density matrices of possibly mixed states and $\psi$ if the state is pure.
All states are assumed to be normalized $\Tr[\rho] = 1$.
Their reduced states on the bath and subsystem are denoted using superscript letters like in $\rho^B =\Tr_S[\rho]$ and $\rho^S =\Tr_B[\rho]$.
We write the trace norm of a density matrix $\rho$ as $\|\rho\|_1 = \Tr[\sqrt{\rho^\dagger\,\rho}] = \Tr|\rho|$,
and the trace distance as
\begin{equation}
  \label{eq:definitiontracedistance}
  \tracedistance(\rho,\sigma) = \frac{1}{2} \|\rho -\sigma \|_1 .
\end{equation}
We denote the operator norm of a Hermitian operator $A$ acting on some Hilbert space $\hiH$ by
\begin{equation}
  \|A\|_\infty = \max_{\psi \in \mathcal{P}_1(\hiH)} \Tr[A\,\psi] ,
\end{equation}
where $\mathcal{P}_1(\hiH)$ is the set of rank one projectors on $\hiH$.
We use the letter $\omega$ to denote the time average of time dependent states $\rho_t$
\begin{equation}
  \omega = \expect{\rho_t}_t = \lim_{\tau\to\infty} \frac{1}{\tau} \int_0^\tau \rho_t\,dt .
\end{equation}

\section{Equilibration}
\label{sec:equilibration}
In a time reversal invariant theory equilibration in the usual sense is impossible.
We therefore use an extended notion of equilibration and say that a system is in equilibrium when its density matrix stays close to some state for almost all times, and say that it evolves toward equilibrium if it approaches such a state and then stays close to it if started in a state far from equilibrium.

Recently, it has been shown that under the above assumptions and whenever the initial state has a high \emph{effective dimension} $d^{\mathrm{eff}}(\omega) = 1/\Tr[\omega^2],\ \omega = \expect{\rho_t}_t$ every small subsystem with $d_S \ll d^{\mathrm{eff}}(\omega)$ equilibrates in this extended sense:
\begin{theorem}[\cite{Linden09}]
  \label{theorem:distancefromtimeaverage}
  Consider any pure state $\psi_t$ evolving under a Hamiltonian with non-degenerate energy gaps.
  Then the average distance between $\rho^S_t = \Tr_B\psi_t$ and its time average $\omega^S = \expect{\rho^S_t}_t$ is bounded by
  \begin{equation}
    \expect{\tracedistance(\rho^S_t,\omega^S)}_t \leq \frac{1}{2} \sqrt{\frac{d_S}{d^\mathrm{eff}(\omega^B)}} \leq \frac{1}{2} \sqrt{\frac{d_S^2}{d^\mathrm{eff}(\omega)}}
  \end{equation}
\end{theorem}

The second important result of \cite{Linden09} is that the effective dimension $d^\mathrm{eff}(\omega)$, which is a measure for how many energy eigenstates contribute significantly to the initial state, is large for almost all pure states drawn according to the unitary in invariant Haar measure:
\begin{theorem}[\cite{Linden09}]
  \label{theorem:highaveragedeffectivedimensionisgeneric}
  (i) The average effective dimension $\expect{d^{\mathrm{eff}}(\omega)}_{\psi_0}$, where the average is computed over uniformly random pure initial states $\psi_0 \in \mathcal{P}_1(\hiH_R)$ chosen from a subspace $\hiH_R$ of dimension $d_R$, is such that
  \begin{equation}
    \expect{d^{\mathrm{eff}}(\omega)}_{\psi_0} \geq \frac{d_R}{2} .
  \end{equation}
  (ii) For a random pure initial state $\psi_0 \in \mathcal{P}_1(\hiH_R)$, the probability that $d^{\mathrm{eff}}(\omega)$ is smaller than $d_R/4$ is exponentially small, namely,
  \begin{equation}
    \probability\left\{d^{\mathrm{eff}}(\omega) < \frac{d_R}{4}\right\} \leq 2 \ee^{-C\,\sqrt{d_R}}
  \end{equation}
  with a constant $C = \ln(2)^2/(72\,\pi^3)$.
\end{theorem}

The Haar measure used in the above theorem is sometimes criticized for being unphysical. 
As the results presented herein depend crucially on $d^{\mathrm{eff}}$ being large it is therefore worth saying a few words about why we believe in their physical significance despite the criticism concerning the Haar measure:
first theorem~\ref{theorem:highaveragedeffectivedimensionisgeneric} is a measure theoretic result, it is not to be misunderstood as statement about states drawn from an actual physical ensemble.
The bound on the probability to get a state with a low effective dimension drops off exponentially.
This raises the hope that the result does not depend on the details of the measure from which the states are sampled and that similar statements can be proven for other non-singular measures.
Second, theorem~\ref{theorem:highaveragedeffectivedimensionisgeneric} is a very strong statement and what is actually needed in the following is much weaker, namely that $d^{\text{eff}}(\omega)$ is much larger than some low, fixed power of $d_S$, which can be as low as 4 or 8 for a single Qubit.
It seems to be unreasonable to assume that the quantum state of a macroscopic object is composed of only that few energy eigenstates.

\section{Speed of fluctuations around equilibrium}
\label{sec:speedoffluctuationsaroudequilibrium}
Knowing that, under suitable conditions, subsystems of large quantum mechanical systems will equilibrate, it is natural to ask: how fast will the fluctuations around the equilibrium state typically be? 
This question was investigated very recently in \cite{0907.1267v1}.

The first step is to introduce a meaningful notion of \emph{speed}.
This is achieved by defining the time derivative \cite{0907.1267v1}
\begin{equation}
  v_S(t) = \lim_{\delta t \to 0} \frac{\tracedistance(\rho^S_t,\rho^S_{t+\delta t})}{\delta t} = \frac{1}{2} \left\| \frac{d\rho^S_t}{dt} \right\|_1 ,
\end{equation}
with 
\begin{equation}
  \frac{d\rho^S_t}{dt} = \iu\,\Tr_B[\rho_t,\haH] .
\end{equation}
As the choice of the origin of the energy scale does not influence the speed, it is convenient to split up the Hamiltonian $\haH$ of the system in a part $\haH_0$ proportional to the identity and the traceless operators $\haH_S$, $\haH_B$ and $\haH_{SB}$ as follows:
\begin{equation}
  \label{eq:generalhamiltonianwithconstanttermabsorbed1}
  \haH = \haH_0 + \haH_S \otimes \mathds{1} + \mathds{1} \otimes \haH_B + \haH_{SB}
\end{equation}

Using a result from \cite{Reimann08} it is shown in \cite{0907.1267v1} that:
\begin{theorem}[\cite{0907.1267v1}]
  \label{theorem:averagespeedisslow}
  For every initial state $\rho_0$ of a composite system evolving under a Hamiltonian of the form \eqref{eq:generalhamiltonianwithconstanttermabsorbed1} and with non-degenerate energy gaps, it holds that:
  \begin{equation}
    \expect{v_S(t)}_t \leq \| \haH_S \otimes \mathds{1} + \haH_{SB}\|_\infty \sqrt{\frac{d_S^3}{d^{\mathrm{eff}}(\omega)}}
    \end{equation}
\end{theorem}

As we have argued above, we are convinced that the effective dimension $d^{\mathrm{eff}}(\omega)$ is typically very large in realistic thermodynamic systems.
In particular, as all dimensions grow exponentially with the number of constituents of the system it will usually be much larger than any fixed power of $d_S$.
Therefore, the speed of the subsystem will, most of the time, be much smaller than $\| \haH_S \otimes \mathds{1} + \haH_{SB}\|_\infty$, which in turn can be expected to grow at most polynomial with the number of constituents of the subsystem and is the natural unit in which the speed of $\rho^S$ is to be measured \cite{0907.1267v1}.

\section{Einselection in a nutshell}
\label{sec:einselectioninanutshell}
The term \emph{einselection}, which stands for \emph{environment-induced super selection}, is due to Zurek \cite{PhysRevD.26.18,RevModPhys.75.715}.
Einselection is known to occur in situations where the Hamiltonian of the composite system leaves a certain orthonormal basis of the subsystem, spanned by so called \emph{pointer states} $\ket{p}$, invariant \cite{Hornberger09}.
If this is the case, the Hamiltonian and the time evolution operator have the form
\begin{align}
  \label{eq:einselectionhamiltonian}
  \haH &= \sum_p \ketbra{p}{p} \otimes \haH^{(p)} \\
  U_t &= \sum_p \ketbra{p}{p} \otimes U^{(p)}_t ,
\end{align}
where $U^{(p)}_t = \ee^{-\iu\,\haH^{(p)}\,t}$ and the $\haH^{(p)}$ are arbitrary Hermitian matrices.
One finds that the subsystem state of an initial product state of the form $\rho_0 = \rho^S_0 \otimes \psi^B_0$, where the state of the bath can be assumed to be pure without loss of generality, evolves into
\begin{equation}
  \rho^S_t = \sum_{pp'} \ketbra{p}{p}\rho^S_0\ketbra{p'}{p'}\,\bra{\psi^B_0}{U^{(p')}_t}^\dagger\,U^{(p)}_t\ket{\psi^B_0}
\end{equation}
Under the evolution induced by such a Hamiltonian the diagonal entries of $\rho^S_0$, when expressed in the pointer basis, remain unchanged while the off-diagonal entries are suppressed by a factor of $\bra{\psi^B_0}{U^{(p')}_t}^\dagger\,U^{(p)}_t\ket{\psi^B_0} \leq 1$.
The actual time development of the $\bra{\psi^B_0}{U^{(p')}_t}^\dagger\,U^{(p)}_t\ket{\psi^B_0}$ depends on the explicit model under consideration, but for many models they have been found to decrease rapidly over short time scales \cite{Zeh96,Breuer02,Hornberger09,RevModPhys.75.715,PhysRevD.26.18}.
If some of the $\haH^{(p)}$ lead to an identical time development for the chosen initial bath state there exist subspaces of $\hiH_S$ in which coherence is preserved.

Note that, the diagonal entries, which survive the decoherence, are completely determined by $\rho^S_0$ and do not depend on the initial state of the bath $\psi^B_0$ at all.
The direct opposite situation is the thermodynamic case where the final state is completely determined by the properties of the bath.
Most realistic situations surely lie between these two extremes.

\section{Equilibration and einselection}
\label{sec:equilibrationandeinselection}
Using the results discussed in the preceding sections it is possible to get rid of the quite limiting assumptions on the Hamiltonian and to show that einselection, i.e. decoherence with respect to a fixed basis, naturally occurs in situations where the interaction Hamiltonian $\haH_{SB}$ is weak.

According to \cite{0907.1267v1} the velocity can be written as
\begin{equation}
  \frac{d\rho^S_t}{dt} = \sum_{k=1}^{d_S^2} c_k(t)\,e_k
\end{equation}
where the $d_S^2$ operators $e_k$ form a Hermitian orthonormal basis for $\hiH_S$ such that $\Tr[e_k\,e_l] = \delta_{kl}$ and
\begin{equation}
  c_k(t) = \Tr\big[\rho(t)\,\iu\,[\haH_S \otimes \mathds{1} + \haH_{SB},e_k\otimes \mathds{1}]\big] .
\end{equation}
The velocity depends on $\haH_B$ only implicitly through the trajectory $\rho_t$, but for an arbitrary fixed state $\rho$ the velocity is solely determined by $\haH_S$ and $\haH_{SB}$:
\begin{equation}
  \label{eq:subsystemstatevelocity}
  \frac{d\rho^S}{dt} = \iu\,[\rho^S,\haH_S] + \iu\,\Tr_B[\rho,\haH_{SB}]
\end{equation}
Now if $\haH_{SB}$ is much weaker than $\haH_S$, \eqref{eq:subsystemstatevelocity} is dominated by the first term.
Consequently, the system can only become slow when $[\rho^S,\haH_S]$ is small.
To see when this happens we first establish a general lower bound on the norm of commutators between states and arbitrary Hermitian matrices:
\begin{lemma}
  \label{lemma:lowerboundonnormsofcommuators}
  Let $\rho$ be a normalized state and $A$ a Hermitian observable with eigenvalues $a_k$ and eigenvectors $\ket{a_k}$, then
  \begin{align}
    \label{eq:commutatorbound}
    \| [\rho,A] \|_1 = \| \iu\,[\rho,A] \|_1 &\geq 2\max_{\{(k,l)\}} \sum_{(k,l)} |a_k - a_l|\,|\rho_{kl}|\\
    &\geq 2\max_{kl} |a_k - a_l|\,|\rho_{kl}| .
  \end{align}
  where the maximization is performed over all decompositions of the index set $\{1,\dots,d_S\}$ into non overlapping pairs $(k,l)$ over which the sum is performed and $\rho_{kl} = \bra{a_k}\rho\ket{a_l}$.
\end{lemma}
\begin{proof}
  The equality is trivial.
  For all traceless, Hermitian, bounded operators $B$ on some finite dimensional Hilbert space $\hiH$ it holds that $\|B\|_1 = 2\,\max_{\Pi \in \mathcal{P}(\hiH)} \Tr[\Pi\,B]$, 
  where $\mathcal{P}(\hiH)$ is the set of all projectors on $\hiH$ and the maximum is obtained when $\Pi$ is the projector onto the positive subspace of $B$.
  By expanding $\rho$ in the eigenbasis of $A$, using the above equality for $B=[\rho,A]$ and considering all sums of orthogonal rank one projectors $\Pi_{kl}$ of the form
  \begin{equation}
    \Pi_{kl} = \ketbra{\pi_{kl}}{\pi_{kl}} \qquad \ket{\pi_{kl}} = \frac{1}{\sqrt{2}}(\ket{a_k} + \ee^{\iu \phi_{kl}} \ket{a_l}) ,
  \end{equation}
  where $\phi_{kl}$ are phase factors, one easily verifies \eqref{eq:commutatorbound}.
  The second inequality is trivial.
\end{proof}

We can now prove the main result of this paper:
\begin{theorem}
  \label{theorem:slowstatesmustdecohere}
  Consider a physical system evolving under a Hamiltonian of the form given in \eqref{eq:generalhamiltonianwithconstanttermabsorbed1} and with non-degenerate energy gaps.
  All reduced states $\rho^S$ satisfy
  \begin{align}
    \|\haH_{SB}\|_{\infty} + \frac{1}{2}\left\|\frac{d\rho^S}{dt}\right\|_1 &\geq \max_{\{(k,l)\}} \sum_{(k,l)} |E^S_k - E^S_l|\,|\rho^S_{kl}|\\
    &\geq \max_{kl} |E^S_k - E^S_l|\,|\rho^S_{kl}| ,
  \end{align}
  where $\rho^S_{kl} = \bra{E^S_k} \rho^S \ket{E^S_l}$ and $E^S_k$ and $\ket{E^S_k}$ are the eigenvalues and eigenstates of $\haH_S$.
\end{theorem}
\begin{proof}
  Using the inverse triangle inequality and \eqref{eq:subsystemstatevelocity} we see that
  \begin{equation}
    | \| \iu\,[\rho^S,\haH_S] \|_1 - \| \iu\,\Tr_B[\rho,\haH_{SB}] \|_1 | \leq \left\|\frac{d\rho^S}{dt}\right\|_1 .
  \end{equation}
  For $\|d\rho^S/dt\|_1$ to become small the norms of the two commutators must be approximately equal.
  Applying lemma~\ref{lemma:lowerboundonnormsofcommuators} to the norm of the first commutator yields:
  \begin{equation}
    \| \iu\,[\rho^S,\haH_S] \|_1 \geq 2\,\max_{k \neq l} |E^S_k - E^S_l|\,|\rho^S_{kl}|
  \end{equation}
  The norm of the second commutator can be upper bounded, using the well-known fact that the trace norm of traceless, Hermitian matrices is non-increasing under completely positive, Hermitian, trace-non-increasing maps as follows:
  \begin{equation}
    \| \iu\,\Tr_B[\rho,\haH_{SB}] \|_1 \leq \| [\rho,\haH_{SB}] \|_1 \leq 2\,\|\haH_{SB}\|_{\infty}
  \end{equation}
\end{proof}

The assertion of theorem~\ref{theorem:slowstatesmustdecohere} is almost intuitively clear, but combined with theorem~\ref{theorem:averagespeedisslow} it allows to draw the following important conclusion:
Whenever $d^{\mathrm{eff}}(\omega)$ is large the subsystem is slow most of the time and if this is the case coherent superpositions of eigenstates of $\haH_S$ with eigenvalue differences that are much larger than $\|\haH_{SB}\|_\infty$ can not contribute significantly to the state of the subsystem.
That is, the corresponding off-diagonal elements of the reduced state $\rho^S_t$ in the $\haH_S$ eigenbasis must be small.
A similar behavior was observed for a specific model in \cite{PhysRevLett.82,Wang08}.

Without using any approximations we have shown that coherence can only be retained between eigenstates of $\haH_S$ whose energy difference is small compared to $\|\haH_{SB}\|_\infty$.
This statement remains meaningful even when the subsystem is large and its energy spectrum thus very dense.
Theorem~\ref{theorem:slowstatesmustdecohere} then still implies that coherent superpositions of eigenstates with far apart energies (sometimes called Sch\"{o}dinger cat states) must decohere.
If the subsystem is small and the interaction Hamiltonian weak compared to the energy gaps of the subsystem Hamiltonian it implies an even stronger statement.
The state of the subsystem must then, most of the time, be approximately diagonal in the eigenbasis of $\haH_S$.

\section{Conclusions}
\label{sec:conclusions}
We have shown that quantum systems which interact weakly with the environment tend to evolve into convex combinations of energy eigenstates of their Hamiltonian.
This result is obtained without making any approximations and without neglecting memory effects in the bath.
No assumptions on the details of the interaction are made, as opposed to the classical einselection mechanism due to Zurek.
In particular, we do not need to assume that the interaction with the environment leaves a certain set of pure pointer states invariant.
This proves that decoherence with respect to a fixed basis is a very natural property of weakly interacting quantum systems.
Due to the generality of our approach we cannot say much about the time scales on which decoherence happens.
For this, specific models must be considered \cite{Wang08,Cramer09}.
Our result establishes a link between decoherence theory and the recent research on equilibration and the foundations of Statistical Mechanics \cite{Linden09,0907.1267v1}.

Decoherence in the energy eigenbasis is observed in many situations where the local Hamiltonian is much stronger than the interaction.
A well-known example is electronic excitations of gases at moderate temperature.
The energy gaps between the ground state and the first few excited states are typically much larger than the thermal energy.
The dynamics of such systems is successfully described using transition rates between energy eigenstates.
Ultimately theorem~\ref{theorem:slowstatesmustdecohere} explains why this is eligible.

\section{Acknowledgments}
\label{sec:acknowledgments}
The author would like to thank Andreas Winter for introducing him to the subject and the ongoing support, and Haye Hinrichsen, Peter Janotta, and Alexander Streltsov for the discussions and useful comments.



\begin{thebibliography}{17}
\expandafter\ifx\csname natexlab\endcsname\relax\def\natexlab#1{#1}\fi
\expandafter\ifx\csname bibnamefont\endcsname\relax
  \def\bibnamefont#1{#1}\fi
\expandafter\ifx\csname bibfnamefont\endcsname\relax
  \def\bibfnamefont#1{#1}\fi
\expandafter\ifx\csname citenamefont\endcsname\relax
  \def\citenamefont#1{#1}\fi
\expandafter\ifx\csname url\endcsname\relax
  \def\url#1{\texttt{#1}}\fi
\expandafter\ifx\csname urlprefix\endcsname\relax\def\urlprefix{URL }\fi
\providecommand{\bibinfo}[2]{#2}
\providecommand{\eprint}[2][]{\url{#2}}

\bibitem[{\citenamefont{Bassi and Ghirardi}(2003)}]{Bassi03}
\bibinfo{author}{\bibfnamefont{A.}~\bibnamefont{Bassi}} \bibnamefont{and}
  \bibinfo{author}{\bibfnamefont{G.}~\bibnamefont{Ghirardi}},
  \bibinfo{journal}{Physics Reports} \textbf{\bibinfo{volume}{379}},
  \bibinfo{pages}{257} (\bibinfo{year}{2003}).

\bibitem[{\citenamefont{Joos et~al.}(1996)\citenamefont{Joos, Zeh, Kiefer,
  Giulini, Kupsch, and Stamatescu}}]{Zeh96}
\bibinfo{author}{\bibfnamefont{E.}~\bibnamefont{Joos}},
  \bibinfo{author}{\bibfnamefont{H.}~\bibnamefont{Zeh}},
  \bibinfo{author}{\bibfnamefont{C.}~\bibnamefont{Kiefer}},
  \bibinfo{author}{\bibfnamefont{D.}~\bibnamefont{Giulini}},
  \bibinfo{author}{\bibfnamefont{J.}~\bibnamefont{Kupsch}}, \bibnamefont{and}
  \bibinfo{author}{\bibfnamefont{I.-O.} \bibnamefont{Stamatescu}},
  \emph{\bibinfo{title}{Decoherence and the Appearance of a Classical World in
  Quantum Theory}} (\bibinfo{publisher}{Springer}, \bibinfo{year}{1996}).

\bibitem[{\citenamefont{Breuer and Petruccione}(2002)}]{Breuer02}
\bibinfo{author}{\bibfnamefont{H.-P.} \bibnamefont{Breuer}} \bibnamefont{and}
  \bibinfo{author}{\bibfnamefont{F.}~\bibnamefont{Petruccione}},
  \emph{\bibinfo{title}{The Theory of Open Quantum Systems}}
  (\bibinfo{publisher}{Oxford University Press}, \bibinfo{year}{2002}).

\bibitem[{\citenamefont{Zurek}(2003)}]{RevModPhys.75.715}
\bibinfo{author}{\bibfnamefont{W.~H.} \bibnamefont{Zurek}},
  \bibinfo{journal}{Rev. Mod. Phys.} pp. \bibinfo{pages}{715--775}
  (\bibinfo{year}{2003}).

\bibitem[{\citenamefont{Reimann}(2008)}]{Reimann08}
\bibinfo{author}{\bibfnamefont{P.}~\bibnamefont{Reimann}},
  \bibinfo{journal}{Phys. Rev. Lett.} \textbf{\bibinfo{volume}{101}},
  \bibinfo{pages}{190403} (\bibinfo{year}{2008}).

\bibitem[{\citenamefont{Linden et~al.}(2009{\natexlab{a}})\citenamefont{Linden,
  Popescu, Short, and Winter}}]{Linden09}
\bibinfo{author}{\bibfnamefont{N.}~\bibnamefont{Linden}},
  \bibinfo{author}{\bibfnamefont{S.}~\bibnamefont{Popescu}},
  \bibinfo{author}{\bibfnamefont{A.~J.} \bibnamefont{Short}}, \bibnamefont{and}
  \bibinfo{author}{\bibfnamefont{A.}~\bibnamefont{Winter}},
  \bibinfo{journal}{Physical Review E} \textbf{\bibinfo{volume}{79}},
  \bibinfo{pages}{061103} (\bibinfo{year}{2009}{\natexlab{a}}),
  \eprint{0812.2385v1}.

\bibitem[{\citenamefont{Linden et~al.}(2009{\natexlab{b}})\citenamefont{Linden,
  Popescu, Short, and Winter}}]{0907.1267v1}
\bibinfo{author}{\bibfnamefont{N.}~\bibnamefont{Linden}},
  \bibinfo{author}{\bibfnamefont{S.}~\bibnamefont{Popescu}},
  \bibinfo{author}{\bibfnamefont{A.~J.} \bibnamefont{Short}}, \bibnamefont{and}
  \bibinfo{author}{\bibfnamefont{A.}~\bibnamefont{Winter}}
  (\bibinfo{year}{2009}{\natexlab{b}}), \eprint{0907.1267v1}.

\bibitem[{\citenamefont{Popescu et~al.}(2006)\citenamefont{Popescu, Short, and
  Winter}}]{Popescu06}
\bibinfo{author}{\bibfnamefont{S.}~\bibnamefont{Popescu}},
  \bibinfo{author}{\bibfnamefont{A.~J.} \bibnamefont{Short}}, \bibnamefont{and}
  \bibinfo{author}{\bibfnamefont{A.}~\bibnamefont{Winter}},
  \bibinfo{journal}{Nature Physics} \textbf{\bibinfo{volume}{2}},
  \bibinfo{pages}{754} (\bibinfo{year}{2006}).

\bibitem[{\citenamefont{Gogolin}(2010)}]{Gogolin10}
\bibinfo{author}{\bibfnamefont{C.}~\bibnamefont{Gogolin}},
  \emph{\bibinfo{title}{Pure state quantum statistical mechanics}}
  (\bibinfo{year}{2010}), \eprint{1003.5058v1},
  \urlprefix\url{http://www.citebase.org/abstract?id=oai:arXiv.org:1003.5058}.

\bibitem[{\citenamefont{Benatti and Floreanini}(2006)}]{Benatti06}
\bibinfo{author}{\bibfnamefont{F.}~\bibnamefont{Benatti}} \bibnamefont{and}
  \bibinfo{author}{\bibfnamefont{R.}~\bibnamefont{Floreanini}},
  \bibinfo{journal}{J. Phys. A} \textbf{\bibinfo{volume}{39}},
  \bibinfo{pages}{2689} (\bibinfo{year}{2006}).

\bibitem[{\citenamefont{Devi and Rajagopal}(2009)}]{Devi09}
\bibinfo{author}{\bibfnamefont{A.~R.~U.} \bibnamefont{Devi}} \bibnamefont{and}
  \bibinfo{author}{\bibfnamefont{A.~K.} \bibnamefont{Rajagopal}},
  \bibinfo{journal}{Phys. Rev. E} \textbf{\bibinfo{volume}{80}}
  (\bibinfo{year}{2009}).

\bibitem[{\citenamefont{Von~Neumann}(1929)}]{vonneumann1929}
\bibinfo{author}{\bibfnamefont{J.}~\bibnamefont{Von~Neumann}},
  \bibinfo{journal}{Zeitschrift f{\"{u}}r Physik A}
  \textbf{\bibinfo{volume}{57}}, \bibinfo{pages}{30} (\bibinfo{year}{1929}),
  \eprint{1003.2133v1}.

\bibitem[{\citenamefont{Zurek}(1982)}]{PhysRevD.26.18}
\bibinfo{author}{\bibfnamefont{W.~H.} \bibnamefont{Zurek}},
  \bibinfo{journal}{Phys. Rev. D} \textbf{\bibinfo{volume}{26}},
  \bibinfo{pages}{1862} (\bibinfo{year}{1982}).

\bibitem[{\citenamefont{Hornberger}(2009)}]{Hornberger09}
\bibinfo{author}{\bibfnamefont{K.}~\bibnamefont{Hornberger}},
  \bibinfo{journal}{Lect. Notes Phys.} \textbf{\bibinfo{volume}{768}},
  \bibinfo{pages}{223} (\bibinfo{year}{2009}), \eprint{quant-ph/0612118v3}.

\bibitem[{\citenamefont{Paz and Zurek}(1999)}]{PhysRevLett.82}
\bibinfo{author}{\bibfnamefont{J.~P.} \bibnamefont{Paz}} \bibnamefont{and}
  \bibinfo{author}{\bibfnamefont{W.~H.} \bibnamefont{Zurek}},
  \bibinfo{journal}{Phys. Rev. Lett.} \textbf{\bibinfo{volume}{82}},
  \bibinfo{pages}{5181} (\bibinfo{year}{1999}).

\bibitem[{\citenamefont{Wang et~al.}(2008)\citenamefont{Wang, Gong, Casati, and
  Li}}]{Wang08}
\bibinfo{author}{\bibfnamefont{W.-g.} \bibnamefont{Wang}},
  \bibinfo{author}{\bibfnamefont{J.}~\bibnamefont{Gong}},
  \bibinfo{author}{\bibfnamefont{G.}~\bibnamefont{Casati}}, \bibnamefont{and}
  \bibinfo{author}{\bibfnamefont{B.}~\bibnamefont{Li}},
  \bibinfo{journal}{Physical Review A} \textbf{\bibinfo{volume}{77}},
  \bibinfo{pages}{012108} (\bibinfo{year}{2008}).

\bibitem[{\citenamefont{Cramer and Eisert}(2009)}]{Cramer09}
\bibinfo{author}{\bibfnamefont{M.}~\bibnamefont{Cramer}} \bibnamefont{and}
  \bibinfo{author}{\bibfnamefont{J.}~\bibnamefont{Eisert}}
  (\bibinfo{year}{2009}), \eprint{0911.2475v1}.

\end{thebibliography}

\end{document}